\documentclass[12pt,reqno]{amsart}

\theoremstyle{plain}
\newtheorem{theorem}{Theorem}[section]
\newtheorem{lemma}[theorem]{Lemma}

\theoremstyle{definition}

\newtheorem{remark}[theorem]{Remark}

\begin{document}

\numberwithin{equation}{section}

\baselineskip=15.5pt

\title[Quot scheme and deformation quantization]{Quot scheme and deformation quantization}

\author[I. Biswas]{Indranil Biswas}

\address{Department of Mathematics, Shiv Nadar University, NH91, Tehsil Dadri,
Greater Noida, Uttar Pradesh 201314, India}

\email{indranil.biswas@snu.edu.in, indranil29@gmail.com}

\subjclass[2010]{53D55, 81S10, 14H60, 14D21}

\keywords{Deformation quantization, projective structure, quot scheme, Moyal--Weyl deformation quantization.}

\date{}

\begin{abstract}
Let $X$ be a compact connected Riemann surface, and let ${\mathcal Q}(r,d)$ denote the quot scheme parametrizing
the torsion quotients of ${\mathcal O}^{\oplus r}_X$ of degree $d$. Given a projective structure $P$ on
$X$, we show that the cotangent bundle $T^*{\mathcal U}$ of a certain nonempty Zariski open subset
${\mathcal U}\, \subset\, {\mathcal Q}(r,d)$, equipped with the natural Liouville symplectic form, admits
a canonical deformation quantization. When $r\,=\,1\,=\, d$, then ${\mathcal Q}(r,d)\,=\, X$; this case was
addressed earlier in \cite{BB}.
\end{abstract}

\maketitle

\section{Introduction}

Let $Y$ be a complex manifold and $\omega$ a holomorphic symplectic form on $Y$. Then $\omega$
defines a Poisson structure on $Y$. Let ${\mathcal A}_0(Y)$ denote the sheaf of locally
defined holomorphic functions on $Y$, and define ${\mathcal A}(Y)\,:=\, {\mathcal A}_0(Y)[[h]]$,
where $h$ is a formal parameter. A deformation quantization of the Poisson structure associated
to $\omega$ is an associative algebra operation $\star$ on ${\mathcal A}(Y)$
for which the following conditions hold:

For any $\widetilde{f}\,:=\, \sum_{i=0}^\infty f_i h^i$ and $\widetilde{g}\, :=\,
\sum_{i=0}^\infty g_ih^i\, \in\, {\mathcal A}(Y)$,
the product
$$
\widetilde{f}\star \widetilde{g} \, =\, \sum_{i=0}^\infty \psi_i h^i 
$$
satisfies the following four conditions:
\begin{itemize}
\item{} each $\psi_i\, \in\, {\mathcal A}_0(Y)$
is some polynomial (independent of $\widetilde f$ and $\widetilde g$)
in derivatives (of arbitrary order) of $\{f_i\}_{i\geq 0}$ and
$\{g_i\}_{i\geq 0}$,
\item{} $\psi_0\, =\, f_0g_0$,
\item{} $1\star f_0\, =\, f_0\star 1\, =\,f_0$ for every $f_0\,\in\,
{\mathcal A}_0(Y)$, and
\item{} $\widetilde{f}\star \widetilde{g}\, -\, \widetilde{g}\star\widetilde{f} \, =\,
\sqrt{-1}\{f_0,\, g_0\}h \,
+\, \beta h^2$, where $\beta\,\in\, {\mathcal A}(Y)$ depends on $f,\,g$.
\end{itemize}
(See Section \ref{sed1} and Section \ref{sed2} for more details and references.)

Let $X$ be a compact connected Riemann surface. Fix positive integers $r$ and $d$. Let
$$
{\mathcal Q}\,:=\, {\mathcal Q}(r,d)
$$
be the quot scheme that parametrizes all torsion quotients of ${\mathcal O}^{\oplus r}_X$ of degree $d$.
It is an irreducible smooth complex projective variety of dimension $rd$. For any Zariski open subset
${\mathcal U}'\, \subset\, T^*{\mathcal Q}$, consider the Liouville form $\theta_{{\mathcal U}'}$ on it
obtained by restricting the Liouville symplectic form on $T^*{\mathcal Q}$; it is a holomorphic symplectic form
on ${\mathcal U}'$.

We show that ${\mathcal U}\, \subset\, T^*{\mathcal Q}$ equipped with the holomorphic symplectic form
$\theta_{\mathcal U}$, where ${\mathcal U}\, \subset\, T^*{\mathcal Q}$ is a certain Zariski open subset,
has a natural deformation quantization, once a projective structure on $X$ is fixed; see Theorem
\ref{thm1}.

We recall that a projective structure on $X$ is given by a covering of $X$ by holomorphic coordinate
charts such that all the transition functions are M\"obius transformations. Using the uniformization
theorem it is easy to see that every Riemann surface admits a projective structure.

When $r\,=\,1\,=\,d$, then we have ${\mathcal Q}(r,d)\,=\, X$. This isomorphism sends any torsion
sheaf to its support. Let ${\mathcal U}_X \,=\, T^*X \setminus 0_X$ be the complement of the image
of the zero section in the total space $T^*X$ of the holomorphic cotangent bundle of $X$. Fix
a projective structure on $X$. In \cite{BB} it was shown that ${\mathcal U}_X$ equipped with the Liouville
symplectic form admits a natural deformation quantization.

\section{A quot scheme}\label{se1}

Let $X$ be an irreducible smooth complex projective curve of genus $g$. Let $V$ be a holomorphic vector
bundle on $X$ of rank $r$. Fix a positive integer $d$. Consider all torsion quotients of degree $d$
of the coherent sheaf $V$. So such a quotient $T$ fits in a short exact sequence
\begin{equation}\label{e0}
0\, \longrightarrow\, W \,:=\, \text{kernel}(q_{_T})\, \longrightarrow\, V \,
\stackrel{q_{_T}}{\longrightarrow}\, T \, \longrightarrow\, 0\, ,
\end{equation}
where $q_{_T}$ is the quotient map and $W$ is a holomorphic vector bundle on $X$ with
${\rm rank}(W)\,=\, r$ and $\text{degree}(W)\,=\, \text{degree}(V)-d$.

There is a projective scheme ${\mathcal Q}\,=\, {\mathcal Q}(V, d)$ defined over $\mathbb C$ that parametrizes
all such quotients. There is also a universal quotient sheaf on $X\times {\mathcal Q}$
\begin{equation}\label{e1}
p^*_{_X} V\, \longrightarrow\, {\mathcal T}\,\longrightarrow\, 0\, ,
\end{equation}
where $p_{_X}$ is the natural projection of $X\times {\mathcal Q}$ to $X$ \cite{Gr}, \cite{BGL},
\cite{Bi}. More
precisely, the pair $({\mathcal Q},\, {\mathcal T})$ has the following universal
property: Let $S$ be a scheme over $\mathbb C$, and let
$p^*_1 V \, \longrightarrow\, T_S\, \longrightarrow\, X\times S$ be a quotient 
of $p^*_1 V$, where $p_1$ is the natural projection of $X\times S$ to $X$, such that
\begin{itemize}
\item{} $T_S$ is flat over $S$, and

\item the relative degree of $T_S$ for the projection $p_{_X}$ is $d$.
\end{itemize}
Then there is a unique morphism $\chi_{_S}\, :\, S\, \longrightarrow\, {\mathcal Q}$ such that
the quotient $T_S$ of $p^*_1 V$ coincides with the quotient $(\text{Id}_X\times \chi_{_S})^*{\mathcal T}$ (see
\eqref{e1}) of $(\text{Id}_X\times\chi_{_S})^*p^*_{_X} V\,=\, p^*_1 V$.

Take a point $\underline{z}\, \in\, \mathcal Q$; suppose it corresponds to the quotient $V\, 
\longrightarrow\, T$ (see \eqref{e0}). Then the obstruction to the smoothness of $\mathcal Q$ at 
$\underline{z}$ is given by $H^1(X,\, \text{Hom}(W,\, T))$ (see \eqref{e0} for $W$); see \cite{Gr}, 
\cite{BGL}. Since $T$ is a torsion sheaf, we have $H^1(X,\, \text{Hom}(W,\, T))\,=\, 0$. Consequently,
$\mathcal Q$ is a smooth complex projective variety. It's irreducible and its dimension is $rd$.

Let ${\mathcal V}$ denote the kernel of the quotient homomorphism in \eqref{e1}. This
$\mathcal V$ is an algebraic vector bundle on $X\times \mathcal Q$. For any $q\, \in\,
\mathcal Q$, the restriction ${\mathcal V}\vert_{X\times \{q\}}$ is a holomorphic vector
bundle on $X$ of rank $r$ and degree $\text{degree}(V)-d\,=:\,\delta$. Hence $\mathcal V$
is an algebraic family of vector bundles on $X$, parametrized by $\mathcal Q$, of rank $r$
and degree $\delta$. Therefore,
we have a morphism
\begin{equation}\label{e2}
\varphi_V\, :\, {\mathcal Q}\, \longrightarrow\,{\mathcal M}_X(r, \delta)\, ,
\end{equation}
where ${\mathcal M}_X(r, \delta)$ is the moduli stack of vector bundles on $X$ of rank $r$
and degree $\delta$.

Take a point $x_0\, \in\, X$. The inclusion of the sheaf ${\mathcal O}_X$ in
${\mathcal O}_X(x_0)$ produces an inclusion
$V\, \hookrightarrow\, V^1\, :=\, V\otimes {\mathcal O}_X(x_0)$. If $T$ is a quotient
sheaf of $V$, then there is a unique quotient $T^1$ of $V^1$ such that the following
diagram is commutative:
\begin{equation}\label{e4}
\begin{matrix}
0 & \longrightarrow & V & \longrightarrow & V^1 & \longrightarrow & V^1_{x_0}
& \longrightarrow & 0\\
&& \Big\downarrow && \Big\downarrow && \Vert\\
0 & \longrightarrow & T & \longrightarrow & T^1 & \longrightarrow & V^1_{x_0}
& \longrightarrow & 0
\end{matrix}
\end{equation}
where $V^1_{x_0}$ is the fiber of $V^1$ over $x_0$, and the two rows in \eqref{e4} are exact. In fact,
$T^1\,=\, V^1/\text{ker}_T$, where $\text{ker}_T$ is the kernel of the quotient
homomorphism $V\, \longrightarrow\, T$; for $T$ in \eqref{e0}, we have $\text{ker}_T\,=\, W$.
Sending any quotient $T$ of $V$ to the corresponding quotient $T^1$ of $V^1$, we get an
embedding
\begin{equation}\label{e3}
f_V\, :\, {\mathcal Q}(V, d)\,\longrightarrow\, {\mathcal Q}(V^1, d+r)\, .
\end{equation}
Note that the following diagram of maps is commutative:
$$
\begin{matrix}
{\mathcal Q}(V, d)& \stackrel{\varphi_V}{\longrightarrow} & {\mathcal M}_X(r, \delta)\\
\,~\,\Big\downarrow f_V && ~\Big\downarrow{\rm Id}\\
{\mathcal Q}(V^1, d+r)& \stackrel{\varphi_{V^1}}{\longrightarrow} & {\mathcal M}_X(r, \delta)
\end{matrix}
$$
where $f_V$ is constructed in \eqref{e3}, $\varphi_V$ is constructed in \eqref{e2}, and
$\varphi_{V^1}$ is constructed as in \eqref{e2} for the vector bundle $V^1$.

For each positive integer $n$, define $V^n\,:=\, V\otimes {\mathcal O}_X(nx_0)$. 
Replacing $V$ by $V^n$ in \eqref{e3}, we have an embedding
$$
{\mathcal Q}(V^n, d+nr)\,\hookrightarrow\, {\mathcal Q}(V^{n+1}, d+(n+1)r)\, .
$$
Let
$$
\widetilde{\mathcal Q}\,:=\, \varinjlim_{n} {\mathcal Q}(V^n, d+nr)
$$
be the direct limit constructed using the above embeddings; it is an ind--scheme. In view
of \eqref{e4}, from \eqref{e2} we have a morphism 
$$
\varphi\, :\, \widetilde{\mathcal Q}\, \longrightarrow\,{\mathcal M}_X(r, \delta)\, .
$$
This morphism $\varphi$ is surjective.

\section{A canonical deformation quantization}

\subsection{Deformation quantization}\label{sed1}

Let $Y$ be a connected complex manifold equipped with a holomorphic symplectic form $\omega$.
So $\omega$ is a closed holomorphic two--form on $Y$ such that for every $y\,\in\, Y$, the
$\mathbb C$--linear map
$$
T_y Y\, \longrightarrow\, \Omega_y\,=\, (T_yY)^*\, , \ \ v\, \longmapsto\,
\{v'\,\longmapsto\, \omega(y)(v\, ,v')\,=\, i_v\omega(y)(v')\}
$$
is an isomorphism. Let
\begin{equation}\label{e5}
P\, :\, TY\, \longrightarrow\, \Omega_Y
\end{equation}
be the isomorphism constructed above. Note that since $\omega$ is holomorphic, the
condition that $\omega$ is $d$--closed is equivalent to the condition that it is
$\partial$--closed.

The symplectic from $\omega$ defines a Poisson structure, which can be described as follows.
Let $f_1$ and $f_2$ be two holomorphic functions defined on an open subset $U$ of $Y$.
Then their Poisson bracket $\{f_1,\, f_2\}$ is defined to be
$$
\{f_1,\, f_2\}\,=\, \omega (P^{-1}(df_1),\, P^{-1}(df_2))\, ,
$$
where $P$ is the isomorphism in \eqref{e5}; note that $\{f_1,\, f_2\}$ is a holomorphic
function on $U$.

Let ${\mathcal A}_0(Y)$ denote the sheaf of locally defined holomorphic functions on $Y$.
Note that the stalks of ${\mathcal A}_0(Y)$ are algebras. Let
$$
{\mathcal A}(Y)\,:=\, {\mathcal A}_0(Y)[[h]]
$$
be the space of all formal Taylor series of the form $\sum_{i=0}^\infty f_i h^i$, where
$f_i\, \in\, {\mathcal A}_0(Y)$ and $h$ is the formal parameter.

A \textit{deformation quantization} of the above Poisson structure is an associative algebra
operation on ${\mathcal A}(Y)$, which is denoted by $\star$,
for which the following conditions hold (see
\cite{BFFLS}, \cite{DL}, \cite{We} for the details):

For any $\widetilde{f}\,:=\, \sum_{i=0}^\infty f_i h^i$ and $\widetilde{g}\, :=\,
\sum_{i=0}^\infty g_ih^i\, \in\, {\mathcal A}(Y)$,
the product
$$
\widetilde{f}\star \widetilde{g} \, =\, \sum_{i=0}^\infty \psi_i h^i 
$$
satisfies the following four conditions:
\begin{itemize}
\item{} each $\psi_i\, \in\, {\mathcal A}_0(Y)$
is some polynomial (independent of $\widetilde f$ and $\widetilde g$)
in derivatives (of arbitrary order) of $\{f_i\}_{i\geq 0}$ and
$\{g_i\}_{i\geq 0}$,
\item{} $\psi_0\, =\, f_0g_0$,
\item{} $1\star f_0\, =\, f_0\star 1\, =\,f_0$ for every $f_0\,\in\,
{\mathcal A}_0(Y)$, and
\item{} $\widetilde{f}\star \widetilde{g}\, -\, \widetilde{g}\star\widetilde{f} \, =\,
\sqrt{-1}\{f_0,\, g_0\}h \,
+\, \beta h^2$, where $\beta\,\in\, {\mathcal A}(Y)$ depends on $f,\,g$.
\end{itemize}

\subsection{Moyal--Weyl deformation quantization of a symplectic vector space}\label{sed2}

Fix a complex vector space $V$ of complex dimension $2d$. Fix
a constant symplectic form $\theta_0$ on $V$. Therefore, $\theta_0\,\in \, \bigwedge^2V^*$,
and the element $\bigwedge^d \theta_0$ of $\bigwedge^{2d}V^*$ is nonzero. As before, let
${\mathcal A}_0(V)$ denote the sheaf of all locally defined holomorphic
functions on $V$. It is equipped with the holomorphic Poisson structure given by $\theta_0$.

Let
$$
\Delta\, :\, V\, \longrightarrow\, V\times V
$$
denote the diagonal homomorphism defined by $v\, \longmapsto
\,(v,\, v)$. There exists a unique differential operator
\begin{equation}\label{diff-op}
D\, : \, {\mathcal A}_0(V \times V) \, \longrightarrow \,
{\mathcal A}_0(V \times V)
\end{equation}
with constant coefficients such that for any pair $f_1\, ,f_2\, \in\,
{\mathcal A}_0(V)$, we have
$$
\{f_1,\, f_2 \} \, =\, \Delta^* D(f_1 \otimes f_2)\, ,
$$
where $f_1 \otimes f_2$ is the function on $V\times V$ defined
by $(u,\, v) \, \longmapsto\, f_1(u) f_2(v)$ \cite{We}
and \cite{Fe}.

The \textit{Moyal--Weyl algebra} is defined by
$$
f_1\star f_2\, =\, \Delta^*\exp (\sqrt{-1}hD/2)(f_1\otimes f_2)
\,\in\, {\mathcal A}(V)
$$
for $f\, ,g\,\in\, {\mathcal A}_0(V)$, and
it is extended to a multiplication operation
on ${\mathcal A}(V)$ using the
bilinearity condition with respect to $h$.
In other words, if
$\widetilde{f}\, :=\, \sum_{i=0}^\infty f_ih^i$
and $\widetilde{g}\, :=\,\sum_{i=0}^\infty g_i h^i$ are two elements of
${\mathcal A}(V)$, then
$$
\widetilde{f}\star\widetilde{g}\, =\, \sum_{i,j=0}^\infty h^{i+j} (f_i\star g_j) \,\in\,
{\mathcal A}(V)\, .
$$
It is known that this $\star$ operation makes ${\mathcal A}(V)$ into an
associative algebra that quantizes the symplectic structure $\Theta_0$; see 
\cite{BFFLS} and \cite{We} for the details.

Let $\text{Sp}(V)$ denote the group of all linear automorphisms
of $V$ that preserve the symplectic form $\theta_0$. This group
$\text{Sp}(V)$ acts on ${\mathcal A}(V)$ in an obvious way.
More precisely,
$$
(\sum_{i=0}^\infty f_ih^i)\circ G \, =\,
\sum_{i=0}^\infty (f_i\circ G)h^i\, ,
$$
where $G\, \in\, \text{Sp}(V)$. 
The differential operator $D$ in \eqref{diff-op} evidently
commutes with the diagonal action of $\text{Sp}(V)$
on $V\times V$. This immediately implies that
\begin{equation}\label{symp.-identity}
(\widetilde{f}\circ G)\star (\widetilde{g}\circ G) \, =\, (\widetilde{f}\star\widetilde{g})\circ G\, . 
\end{equation}
for any $G\, \in\, \text{Sp}(V)$ and $\widetilde{f}\, ,\widetilde{g}\,\in\,
{\mathcal A}(V)$.

\section{Projective structures on a Riemann surface}

Let $X$ be a compact connected Riemann surface.

A global holomorphic automorphism of ${\mathbb C}{\mathbb P}^1$ is called
a \textit{M\"obius transformation}. All M\"obius transformations are of the
form $z\, \longmapsto\, \frac{az+b}{cz+d}$ with $ad-bc \,\not=\, 0$. So the group
of M\"obius transformations is identified with the projective
linear group $\text{PGL}(2,{\mathbb C})$.

A holomorphic coordinate function on $X$ is a pair of the form $(U,\, \phi)$, where $U$ is
an Euclidean open subset of $X$, and $\phi\, :\, U\,\longrightarrow\, {\mathbb C}{\mathbb P}^1$
is a biholomorphism onto an open subset of ${\mathbb C}{\mathbb P}^1$.
A collection of holomorphic coordinate functions $\{U_j,\, \phi_j\}_{j\in J}$ is said to cover
$X$ if $\bigcup_{j\in J} U_j\,=\, X$. A projective structure on $X$ is given by a covering of $X$ by a
collection of holomorphic coordinate functions $\{U_j,\, \phi_j\}_{j\in J}$ such that the
transition function $\phi_i\circ \phi^{-1}_j$, when restricted to $\phi_j(U)$, where
$U$ is any connected component of $U_i\bigcap U_j$, coincides with the restriction of
some M\"obius transformation for all $i,\, j\, \in\, J$. (See \cite{Gu}, \cite{De}.)

Any Riemann surface $Z$ admits a projective structure. To see this, note that the uniformization theorem 
says that the universal cover of $Z$ is either $\mathbb C$ or ${\mathbb C}{\mathbb P}^1$ or the 
upper-half plane $\mathbb H$. So the automorphism group of the universal cover of $Z$ is a subgroup of 
$\text{PGL}(2,{\mathbb C})$. Therefore, the natural projective structure on the universal cover of $Z$
given by its inclusion map to ${\mathbb C}{\mathbb P}^1$ preserved by the deck transformations and hence
it descends to a projective structure on $Z$.

The space of all projective structure on
a Riemann surface $Z$ is an affine space for $H^0(Z,\, K^{\otimes 2}_Z)$, where $K_Z\,=\, \Omega_Z$
is the holomorphic cotangent bundle of $Z$.

Let $s_X\, :\, X\, \longrightarrow\, K_X$ be the zero section
The total space of $K_X$ is equipped with the Liouville symplectic form. Define
\begin{equation}\label{e6}
{\mathcal K}_X\,=\, K_X \setminus s_X(X)
\end{equation}
to be the complement of the image of the zero section. Let $\theta_X$ denote the restriction of the 
Liouville symplectic form to ${\mathcal K}_X$. We will show that a projective structure on $X$ produces a 
deformation quantization of the symplectic manifold $({\mathcal K}_X,\, \theta_X)$.

Fix the standard symplectic form $dx\wedge dy$ on ${\mathbb C}^2$, where $(x,\, y)$ are the
standard coordinates. The group ${\mathbb Z}/2\mathbb Z$ acts on ${\mathbb C}^2$ via the involution
$(x,\,y)\, \longmapsto\, (-x,\,-y)$. Let
\begin{equation}\label{y0}
{\mathcal Y}_0\, :=\, ({\mathbb C}^2\setminus \{0\})/({\mathbb Z}/2\mathbb Z)
\end{equation}
be the corresponding quotient. Since the standard symplectic form on ${\mathbb C}^2$ is preserved by the 
action of ${\mathbb Z}/2\mathbb Z$, the complex surface ${\mathcal Y}_0$ gets a holomorphic symplectic 
form, which will be denoted by
\begin{equation}\label{y1}
\theta'.
\end{equation}
{}From \eqref{symp.-identity} it follows immediately that the 
Moyal--Weyl deformation quantization of $({\mathbb C}^2,\, dx\wedge dy)$ produces a deformation 
quantization of the symplectic manifold $({\mathcal Y}_0,\, \theta')$.

Construct ${\mathcal K}_{{\mathbb C}{\mathbb P}^1}$ by setting $X\,=\, {\mathbb C}{\mathbb P}^1$
in \eqref{e6}. The Liouville symplectic form $\theta_{{\mathbb C}{\mathbb P}^1}$ on it
will be denoted by
\begin{equation}\label{y2}
\theta''.
\end{equation}

The standard action of $\text{SL}(2,{\mathbb C})$ on ${\mathbb C}^2$ produces an action of
$\text{SL}(2,{\mathbb C})$ on ${\mathcal Y}_0$. The center $\pm I$ of $\text{SL}(2,{\mathbb C})$ acts
trivially on ${\mathcal Y}_0$, and hence we get an action of $\text{PGL}(2,{\mathbb C})
\,=\, \text{SL}(2,{\mathbb C})/\{\pm I\}$ on
${\mathcal Y}_0$. On the other hand, the action of $\text{PGL}(2,{\mathbb C})$ on
${\mathbb C}{\mathbb P}^1$ produces an action of $\text{PGL}(2,{\mathbb C})$ on
${\mathcal K}_{{\mathbb C}{\mathbb P}^1}$.

\begin{lemma}\label{lem1}
The symplectic surface $({\mathcal Y}_0,\,\theta')$ (see \eqref{y0} and \eqref{y1}) has a natural
holomorphic symplectomorphism with $({\mathcal K}_{{\mathbb C}{\mathbb P}^1},\, \theta'')$ (see
\eqref{y2}). This symplectomorphism is ${\rm PGL}(2,{\mathbb C})$--equivariant.
\end{lemma}

\begin{proof}
Consider the tautological line bundle ${\mathcal O}_{{\mathbb C}{\mathbb P}^1}(-1)$ on
${\mathbb C}{\mathbb P}^1$ whose fiber over any point $p\, \in\, {\mathbb C}{\mathbb P}^1$
is the line in ${\mathbb C}^2$ represented by $p$. Let 
$$s_0\, :\, {\mathbb C}{\mathbb P}^1\, \longrightarrow\, {\mathcal O}_{{\mathbb C}{\mathbb P}^1}(-1)$$
be the zero section. The complement ${\mathbb C}^2\setminus\{0\}$ is identified with the complement
${\mathcal O}_{{\mathbb C}{\mathbb P}^1}(-1)\setminus s_0({\mathbb C}{\mathbb P}^1)$ by sending any
$y\, \in\, {\mathbb C}^2\setminus\{0\}$ to the point in the line ${\mathbb c}\cdot y\, \subset\,
{\mathbb C}^2$ given by $y$.

Let 
$$s'_0\, :\, {\mathbb C}{\mathbb P}^1\, \longrightarrow\, {\mathcal O}_{{\mathbb C}{\mathbb P}^1}(-2)$$
be the zero section. Consider the map of total spaces
$${\mathcal O}_{{\mathbb C}{\mathbb P}^1}(-1)\, \longrightarrow\,
{\mathcal O}_{{\mathbb C}{\mathbb P}^1}(-2)$$ defined by $v\, \longmapsto\, v\otimes v$; note that this
is not fiberwise linear. Using this map, 
the quotient ${\mathcal Y}_0$ in \eqref{y0} is identified with the 
the complement ${\mathcal O}_{{\mathbb C}{\mathbb P}^1}(-2)\setminus s'_0({\mathbb C}{\mathbb P}^1)$.

The holomorphic cotangent bundle $K_{{\mathbb C}{\mathbb P}^1}$ of ${\mathbb C}{\mathbb P}^1$ is identified
with the line bundle ${\mathcal O}_{{\mathbb C}{\mathbb P}^1}(-2)\otimes_{\mathbb C} \left(\bigwedge^2
{\mathbb C}^2\right)^*$. Therefore, using the symplectic form $dx\wedge dy$ on 
${\mathbb C}^2$, the line bundle ${\mathcal O}_{{\mathbb C}{\mathbb P}^1}(-2)$
is identified with $K_{{\mathbb C}{\mathbb P}^1}$. Combining this with the above isomorphism of
${\mathcal Y}_0$ with ${\mathcal O}_{{\mathbb C}{\mathbb P}^1}(-2)\setminus s'_0({\mathbb C}{\mathbb P}^1)$
we obtain a 
holomorphic isomorphism between ${\mathcal Y}_0$ and ${\mathcal K}_{{\mathbb C}{\mathbb P}^1}$. This
biholomorphism takes the symplectic form $\theta'$ on ${\mathcal Y}_0$ to the symplectic form
$\theta''$ on ${\mathcal K}_{{\mathbb C}{\mathbb P}^1}$. Furthermore, this biholomorphism is
${\rm PGL}(2,{\mathbb C})$--equivariant.
\end{proof}

Fix a projective structure $P$ on $X$.

It was shown earlier that the Moyal--Weyl
deformation quantization of $({\mathbb C}^2,\, dx\wedge dy)$ produces a deformation quantization of
the symplectic manifold $({\mathcal Y}_0,\, \theta')$. In view of Lemma \ref{lem1}
and \eqref{symp.-identity}, using the projective structure $P$ on $X$, this deformation quantization of
$({\mathcal Y}_0,\, \theta')$ produces a deformation quantization of 
$({\mathcal K}_X,\,\theta_X)$.

\section{Liouville form for quot scheme}

Let $X$ be a compact connected Riemann surface.

Let $V\,=\, {\mathcal O}^{\oplus r}_X$ be the trivial holomorphic vector bundle on $X$
of rank $r$. The trivial holomorphic line bundle $\bigwedge^r V\,=\, {\mathcal O}_X$ will be
denoted by $\mathcal L$. We have a holomorphic map
\begin{equation}\label{ga}
\gamma\, :\, {\mathcal Q}(V, d)\,\longrightarrow\, {\mathcal Q}({\mathcal L}, d)
\end{equation}
that sends any quotient $q_T\, :\, V\, \longrightarrow\, T$ as in \eqref{e0} to the
quotient ${\mathcal L}/(\bigwedge^r {\rm kernel}(q_{_T})) \,=\, {\mathcal L}/(\bigwedge^r W)$,
where $W$ is the vector bundle in \eqref{e0} defined by ${\rm kernel}(q_{_T})$.

Now, since ${\mathcal L}$ is a line bundle, any torsion quotient $\mathcal T$ of $\mathcal L$ is
uniquely determined by the support of the sheaf $\mathcal T$. Consequently, the quot
scheme ${\mathcal Q}({\mathcal L}, d)$ is identified with the symmetric product
$\text{Sym}^d(X)$ by the map
\begin{equation}\label{e-1}
\eta\,:\, {\mathcal Q}({\mathcal L}, d)\, \longrightarrow\, \text{Sym}^d(X)
\end{equation}
that sends any quotient $\mathcal T$ of $\mathcal L$ to the
support of the sheaf $\mathcal T$ with the scheme structure given by its $0$--th Fitting ideal.

Let $\text{Sym}^d(X)^0\, \subset\, \text{Sym}^d(X)$ be the Zariski open dense subset that
parametrizes all the distinct $d$ points of $X$. In other words, $\text{Sym}^d(X)^0\,
\subset\, \text{Sym}^d(X)$ is the locus of reduced effective divisors on $X$ of degree $d$. Let
\begin{equation}\label{e-2}
\widetilde{\mathcal U}'\, :=\, (\eta\circ\gamma)^{-1}(\text{Sym}^d(X)^0)\, \subset\,
{\mathcal Q}(V, d)
\end{equation}
be the inverse image, where $\gamma$ and $\eta$ are the maps constructed in \eqref{ga} and
\eqref{e-1} respectively, which is a Zariski open dense subset of ${\mathcal Q}(V, d)$. We note that
$$
\widetilde{\mathcal U}'\, =\, \text{Sym}^d(X)^0\times ({\mathbb C}{\mathbb P}^{r-1})^d\, .
$$
Therefore, we have the Zariski open subset
$$
\widetilde{\mathcal U}\, =\, \text{Sym}^d(X)^0\times ({\mathbb C}^{r-1})^d
\,=\, \text{Sym}^d(X)^0\times {\mathbb C}^{d(r-1)}\, \subset\,
\text{Sym}^d(X)^0\times ({\mathbb C}{\mathbb P}^{r-1})^d\, .
$$
So the holomorphic cotangent bundle $T^*\widetilde{\mathcal U}$ is the Cartesian product
\begin{equation}\label{y3}
T^*\widetilde{\mathcal U}\,=\, (T^*\text{Sym}^d(X)^0)\times (T^*{\mathbb C}^{d(r-1)})\, .
\end{equation}
The Liouville holomorphic symplectic form on the cotangent
bundle $T^*\widetilde{\mathcal U}$ will be denoted by
\begin{equation}\label{y4}
\theta_0.
\end{equation}

Note that $T^*\text{Sym}^d(X)^0$ is a Zariski open subset of $\text{Sym}^d(T^*X)$.
On the other hand, the inclusion map ${\mathcal K}_X\, \hookrightarrow\, T^*X\,=\, K_X$ in
\eqref{e6} produces a Zariski open subset
$$
\text{Sym}^d({\mathcal K}_X)\, \subset\, \text{Sym}^d(T^*X)\, .
$$
Let
\begin{equation}\label{cw}
{\mathcal W}\, :=\, (T^*\text{Sym}^d(X)^0)\cap \text{Sym}^d({\mathcal K}_X)\, \subset\,
\text{Sym}^d(T^*X)
\end{equation}
be the intersection, which is a Zariski open subset.

Let
\begin{equation}\label{cu}
{\mathcal U}\,:=\, {\mathcal W}\times T^*{\mathbb C}^{d(r-1)}
\, \subset\, (T^*\text{Sym}^d(X)^0)\times (T^*{\mathbb C}^{d(r-1)})\,=\,
T^*\widetilde{\mathcal U}
\end{equation}
be the Zariski open subset (see \eqref{y3}). The holomorphic symplectic form on $\mathcal U$ obtained by
restricting the Liouville symplectic form $\theta_0$ on $T^*\widetilde{\mathcal U}$ (see \eqref{y4})
will be denoted by $\theta$.

\begin{theorem}\label{thm1}
A projective structure $P$ on $X$ produces a deformation quantization of a Zariski open dense subset
${\mathcal U}\, \subset\, T^*{\mathcal Q}(V, d)$ equipped with the Liouville
symplectic form $\theta$.
\end{theorem}

\begin{proof}
Let $P$ be a projective structure on $X$.
Consider the Liouville symplectic form $\theta_X$ on ${\mathcal K}_X$ (see \eqref{e6}). Using $P$ we
constructed a deformation quantization of $({\mathcal K}_X,\,\theta_X)$. Now we shall construct a
deformation quantization of $\mathcal W$ in \eqref{cw} equipped with the restriction of the Liouville
symplectic form; note that since $\mathcal W$ is an open subset of $T^*\text{Sym}^d(X)^0$, the 
Liouville symplectic form on $T^*\text{Sym}^d(X)^0$ restricts to a symplectic form on
$\mathcal W$.

For $1\, \leq\, i\, \leq\, d$, let $p_i\, :\, ({\mathcal K}_X)^d\, \longrightarrow\,{\mathcal K}_X$
be the projection to the $i$--factor in the Cartesian product. The pull-back
$\sum_{i=1}^d p^*_i\theta_X$ is evidently a holomorphic symplectic form on
${\mathcal K}_X$. A deformation quantization of $({\mathcal K}_X,\,\theta_X)$ produces
a deformation quantization of $(({\mathcal K}_X)^d,\, \sum_{i=1}^d p^*_i\theta_X)$. In particular,
the deformation quantization of $({\mathcal K}_X,\,\theta_X)$ constructed using the
projective structure $P$ on $X$ gives a deformation
quantization of $(({\mathcal K}_X)^d,\, \sum_{i=1}^d p^*_i\theta_X)$.

Let $S_d$ denote the group of permutations of $\{1,\, \cdots,\, d\}$. This group $S_d$ acts on 
$({\mathcal K}_X)^d$ by permuting the factors in the Cartesian product. The corresponding 
quotient space is $\text{Sym}^d({\mathcal K}_X)$. The symplectic form $\sum_{i=1}^d 
p^*_i\theta_X$ on $({\mathcal K}_X)^d$ is evidently preserved by the action of $P^d$. Moreover,
the above deformation quantization of $(({\mathcal K}_X)^d,\, \sum_{i=1}^d p^*_i\theta_X)$ is
also preserved by the action of $P^d$. Therefore, we get a holomorphic symplectic form
$\Theta_d$ on the smooth locus $\text{Sym}^d({\mathcal K}_X)_0\, \subset\,
\text{Sym}^d({\mathcal K}_X)$ given by $\sum_{i=1}^d p^*_i\theta_X$, and we also obtain a
deformation quantization of $(\text{Sym}^d({\mathcal K}_X)_0,\, \Theta_d)$ given by the above
deformation quantization of $(({\mathcal K}_X)^d,\, \sum_{i=1}^d p^*_i\theta_X)$.

The open subset ${\mathcal W}\, \subset\, \text{Sym}^d({\mathcal K}_X)$ in \eqref{cw} is
clearly contained in $\text{Sym}^d({\mathcal K}_X)_0$. Therefore, the above
deformation quantization of $(\text{Sym}^d({\mathcal K}_X)_0,\, \Theta_d)$ produces a
deformation quantization of $({\mathcal W},\, \Theta_d)$ (the restriction of
$\Theta_d$ to ${\mathcal W}$ is also denoted by $\Theta_d$).

Let $\Theta_0$ denote the Liouville symplectic form on $T^*{\mathbb C}^{d(r-1)}$. Since
there is a natural holomorphic symplectomorphism between the holomorphic symplectic
manifold $(T^*{\mathbb C}^{d(r-1)}, \, \Theta_0)$ and ${\mathbb C}^{2d(r-1)}$ equipped with
the standard symplectic form, we have the Moyal--Weyl deformation quantization of
$(T^*{\mathbb C}^{d(r-1)}, \, \Theta_0)$.

The deformation quantizations of $(T^*{\mathbb C}^{d(r-1)}, \, \Theta_0)$ and $({\mathcal 
W},\, \Theta_d)$ together produce a deformation quantizations of
$({\mathcal W}\times T^*{\mathbb C}^{d(r-1)}, \, \Theta_d\oplus \Theta_0)$. The
identification between ${\mathcal W}\times T^*{\mathbb C}^{d(r-1)}$ and $\mathcal U$
(see \eqref{cu}) takes the symplectic form $\Theta_d\oplus \Theta_0$ on
${\mathcal W}\times T^*{\mathbb C}^{d(r-1)}$ to the symplectic form 
$\theta$ on $\mathcal U$. This completes the proof.
\end{proof}

\begin{remark}\label{rem}
For any holomorphic line bundle $L$ on $X$, the quot scheme ${\mathcal Q}(L^{\oplus r}, d)$
is canonically identified with ${\mathcal Q}({\mathcal O}^{\oplus r}_X, d)$. Therefore,
from Theorem \ref{thm1} we conclude that a projective structure on $X$ produces a
deformation quantization of a nonempty Zariski open subset of ${\mathcal Q}(L^{\oplus r}, d)$.
\end{remark} 

\section*{Acknowledgements}

The author is partially supported by the J. C. Bose Fellowship JBR/2023/000003.


\end{document}